\documentclass[twoside,leqno,twocolumn]{article}  
\usepackage{ltexpprt} 

\usepackage{amsmath,amsfonts,amssymb}

\usepackage{bm}

\usepackage{xspace}

\usepackage{paralist}

\usepackage{graphicx}
\DeclareGraphicsExtensions{.png}

\usepackage{subfig}

\usepackage{boxedminipage}

\usepackage[colorlinks,urlcolor=blue,citecolor=blue,linkcolor=blue]{hyperref}

\usepackage{algorithm}
\usepackage{algpseudocode}

\def\submit{1}   
\ifnum\submit=1
	\newcommand{\full}[1]{}
	\newcommand{\confer}[1]{#1}
\else
	\newcommand{\full}[1]{#1}
	\newcommand{\confer}[1]{}
\fi

\newcommand{\myproofend}{\hfill\Box}

\newcommand{\bv}{{\bf v}}

\newcommand{\Sec}[1]{\hyperref[sec:#1]{\S\ref*{sec:#1}}} 
\newcommand{\Eqn}[1]{\hyperref[eq:#1]{(\ref*{eq:#1})}} 
\newcommand{\Fig}[1]{\hyperref[fig:#1]{Fig.\,\ref*{fig:#1}}} 
\newcommand{\Tab}[1]{\hyperref[tab:#1]{Tab.\,\ref*{tab:#1}}} 
\newcommand{\Thm}[1]{\hyperref[thm:#1]{Thm.\,\ref*{thm:#1}}} 
\newcommand{\Lem}[1]{\hyperref[lem:#1]{Lem.\,\ref*{lem:#1}}} 
\newcommand{\Prop}[1]{\hyperref[prop:#1]{Prop.~\ref*{prop:#1}}} 
\newcommand{\Cor}[1]{\hyperref[cor:#1]{Cor.~\ref*{cor:#1}}} 
\newcommand{\Def}[1]{\hyperref[def:#1]{Defn.~\ref*{def:#1}}} 
\newcommand{\Alg}[1]{\hyperref[alg:#1]{Alg.~\ref*{alg:#1}}} 
\newcommand{\Ex}[1]{\hyperref[ex:#1]{Ex.~\ref*{ex:#1}}} 
\newcommand{\Clm}[1]{\hyperref[clm:#1]{Claim~\ref*{clm:#1}}} 

\newcommand{\M}[1]{{{{\MakeUppercase{#1}}}}} 
\newcommand{\qtext}[1]{\quad\text{#1}\quad} 

\begin{document}

\title{The Similarity between Stochastic Kronecker \\and Chung-Lu Graph Models\thanks{This work was funded by the applied mathematics program at the United
States Department of Energy and performed at Sandia National
Laboratories, a multiprogram laboratory operated by Sandia
Corporation, a wholly owned subsidiary of Lockheed Martin Corporation,
for the United States Department of Energy's National Nuclear Security
Administration under contract DE-AC04-94AL85000.} }

\author{Ali Pinar\thanks{Sandia National Laboratories, CA, apinar@sandia.gov} \\
\and 
C. Seshadhri\thanks{Sandia National Laboratories, CA, scomand@sandia.gov}
\and
Tamara G. Kolda\thanks{Sandia National Laboratories, CA, tgkolda@sandia.gov} }
\date{}

\maketitle
\pagestyle{headings}
\setcounter{page}{1}
\pagenumbering{arabic}

\begin{abstract}
The analysis of massive graphs is now becoming a very important part
of science and industrial research. This has led to the construction
of a large variety of graph models, each with their own advantages. 
The \emph{Stochastic Kronecker Graph} (SKG) model has been chosen by
the Graph500 steering committee to create supercomputer benchmarks for 
graph algorithms. The major reasons for this are its easy parallelization
and ability to mirror real data. Although SKG is easy to implement,
there is little understanding of the properties and behavior of this model.

We show that the parallel variant of the edge-configuration model given
by Chung and Lu (referred to as CL) is notably similar to the SKG model. 
The graph properties of an SKG are extremely close to those of a CL
graph generated with the appropriate parameters.
Indeed, the final probability matrix 
used by SKG is almost identical to that of a CL model. This implies that the
graph distribution represented by SKG is almost the same as that given
by a CL model. We also show that when it comes to fitting real data, 
CL performs as well as SKG based on empirical studies of graph
properties. CL has the added benefit of a trivially simple fitting procedure and exactly matching the degree distribution. Our results suggest that
users of the SKG model should consider the CL model because of its similar
properties, simpler structure, and ability to fit a wider range of degree distributions. At the very least, CL is a good control model to compare against.
\end{abstract}

\section{Introduction}
\label{sec:intro}

With more and more data being represented as large graphs, 
network analysis is becoming a major topic of scientific research.
Data that come from social networks, the Web, patent citation networks,
and power grid structures are increasingly being viewed as  massive graphs.
These graphs usually have peculiar properties that distinguish them
from standard random graphs (like those generated from the 
Erd\"{o}s-R\'{e}nyi model). Although we have a lot of evidence
for these properties, we do not have a thorough understanding
of \emph{why} these properties occur. Furthermore, it is not at all
clear how to generate synthetic graphs that have a similar
behavior.

Hence, \emph{graph modeling} is a very important topic of study. 
There may be some disagreement as to the characteristics of a good
model, but the survey \cite{ChFa06} gives a fairly comprehensive list of desired
properties. As we deal with larger and larger graphs, the efficiency
and speed as well as implementation details become deciding
factors in the usefulness of a model. The theoretical benefit
of having a good, fast model is quite clear. However, the benefits of  having good models go  beyond an ability to generate large graphs, since 
such models provide insight into structural properties and the processes that generate large
graphs.

The \emph{Stochastic Kronecker graph} (SKG) \cite{LeFa07,LeChKlFa10}, a generalization
of \emph{recursive matrix} (R-MAT) model \cite{ChZhFa04}, is a model
for large graphs that has received a lot of attention.
It involves few parameters and has an embarrassingly
parallel implementation (so each edge of the graph
can be independently generated). The importance of this model 
cannot be understated --- it has been
chosen to create graphs for the Graph500 supercomputer
benchmark \cite{Graph500}. Moreover,  many researchers generate  SKGs for testing their algorithms~\cite{Vineet:2009, 5713180, Pearce:2010, BulucGilbert, Pothen,  HOO11, Plimpton2011610, 4536261, 4454426,Schmidt2009417}. 

Despite the role of this model in graph benchmarking and algorithm testing, 
precious little is truly known about its properties. The model description is quite
simple, but varying the parameters of the model can have quite drastic
effects on  the properties of the graphs being generated. Understanding \emph{what} goes on while generating
an SKG is extremely difficult. Indeed, merely explaining the
structure of the degree distribution requires a significant amount
of mathematical effort.

Could there be a conceptually simpler model that has properties similar
to SKG? A possible candidate is a simple variant of the Erd\"{o}s-R\'{e}nyi 
model first discussed by Aiello, Chung, and Lu \cite{AiChLu01} and generalized by Chung and Lu \cite{ChLu02, ChLu02-2}. The Erd\"{o}s-R\'{e}nyi model is arguably
the earliest and simplest random graph model \cite{ErRe59, ErRe60}.
The Chung-Lu model (referred to as CL) can be viewed as a version
of the edge configuration model or a weighted Erd\"{o}s-R\'{e}nyi 
graph. 
Given any degree distribution, it generates a random graph with
the same distribution on expectation. 
(The version by Aiello et al.\@ only considered power law distributions.) 
It is very efficient and conceptually very
simple. Amazingly, it has been overlooked as a model to generate
synethetic instances, and is not even considered as a ``control model"
to compare with. (This is probably because of the strong ties to
a standard Erd\"{o}s-R\'{e}nyi graph, which is well known to be
unsuitable for modeling social networks.) A major benefit of this
model is that it can provide graphs with \emph{any} desired degree distribution
(especially power law), something that SKG provably cannot do.

Our aim is to provide a detailed comparison of the SKG and CL models. 
We first compare the graph properties of an SKG graph with an associated
CL graph. We then look at how these models fit real data.
Our observations
show a great deal of similarity between these models.
To explain this, we look directly at the
probability matrices used by these models. This gives insight into the
structure of the graphs generated. We notice that the SKG and CL matrices have
much in common and give evidence that the differences between these 
are only (slightly) quantitative, not qualitative. We also show that for some settings
of the SKG parameters, the SKG and associated CL
models  coincide \emph{exactly}. 
\subsection{Notation and Background}
\label{sec:rmat}

\subsubsection{Stochastic Kronecker Graph (SKG) model} The model takes as input the number of nodes $n$ 
(always a power of $2$), number of edges $m$,
and a $2 \times 2$ generator matrix $\M{T}$. We define $\ell = \log_2 n$ as the number of \emph{levels}.
In theory, the SKG generating matrix can
be larger than $2 \times 2$, but we are unaware of any such examples in practice.
Thus, we assume that the generating matrix has the form
\begin{displaymath}
  \M{T} =
  \begin{bmatrix} 
    t_1 & t_2 \\ 
    t_3 & t_4 
  \end{bmatrix}
  \qtext{with}
  t_1 + t_2 + t_3 + t_4 = 1.
\end{displaymath}
Each edge is inserted according to the probabilities%
\footnote{\scriptsize We have taken a slight liberty in requiring the entries of $T$ to sum to 1. In fact, the SKG model as defined in \cite{LeChKlFa10} works with the matrix $mP$, which is considered the matrix of probabilities for the existence of each individual edge (though it might be more accurate to think of it as an expected value).}
defined by
\begin{displaymath}
  \M{P} = 
  \underbrace{
    \M{T} \otimes \M{T} \otimes \cdots \otimes \M{T}
  }_{\text{$\ell$ times}}.
\end{displaymath}
We will refer to $P_{\textrm{SKG}}$ as the \emph{SKG matrix} associated with these
parameters. Observe that the entries in $P_{\textrm{SKG}}$ sum up to $1$, and hence
it gives a probability distribution over all pairs $(i,j)$. This 
is the probability that a single edge insertion results in the edge
$(i,j)$. By repeatedly using this distribution to generate $m$ edges,
we obtain our final graph.

In practice, the matrix $\M{P_{\textrm{SKG}}}$ is never formed explicitly. 
Instead, each edge is inserted as follows.
Divide the adjacency matrix into
four quadrants, and choose one of them with the corresponding probability $t_1, t_2, t_3$, or $t_4$.
Once a quadrant is chosen, repeat this recursively in that quadrant. Each time we iterate,
we end up in a square submatrix whose dimensions are exactly halved. After $\ell$ iterations,
we reach a single cell of the adjacency matrix, and an edge is inserted. This is independently repeated
$m$ times to generate the final graph.
Note that all edges can be inserted in parallel. This is one of the major advantages
of the SKG model and why it is appropriate for generating large supercomputer benchmarks.

A noisy version of SKG (called NSKG) has been recently designed in \cite{SePiKo11, SePiKo11-arxiv}. This chooses
the probability matrix  
\[\M{P} = \M{T}_1 \otimes \cdots \otimes \M{T_\ell},\] 
 where
each $\M{T}_i$ is a specific random perturbation of the original generator matrix $T$. This
has been provably shown to smooth the degree distribution to a lognormal form.

\subsection{Chung-Lu (CL) model} This model can be thought of as a variant of the edge configuration
model. Let us deal with directed graphs  to describe the CL model. Suppose we are given sequences of $n$ in-degrees
$d_1, d_2, \ldots, d_n$, and $n$ out-degrees $d'_1, d'_2, \ldots, d'_n$. We have
$\sum_i d_i = \sum_i d'_i = m$. Consider the probability matrix $P_{\textrm{CL}}$ where
the $(i,j)$ entry is $d_i d'_j/m^2$. (The sum of all entries in $P_{\textrm{CL}}$ is $1$.)
We use this probability matrix to make $m$ edge insertions.

This is slightly different from the standard CL model, where an independent
coin flip is done for every edge. This is done by using the matrix $mP_{\textrm{CL}}$ 
(similar to SKG). 
In practice, we do not generate $P_{\textrm{CL}}$ explicitly, but have a simple $O(m)$ implementation analogous
to that for SKG. Independently for every edge, we choose a source and a sink. Both of these
are chosen independently using the degree sequences as probability distributions.
This is extremely simple to implement and it is very efficient.

\medskip

We will focus on undirected graphs for the rest of this paper. This is done by performing $m$
edge insertions, and considering each of these to be undirected. For real data that is directed,
we symmetrize by removing the direction. 

Given a set of SKG parameters, we can define the associated CL model. Any set of SKG parameters
immediately defines an \emph{expected degree sequence}. In other words, given the SKG
matrix $P_{\textrm{SKG}}$, we can deduce the expected in-(and out)-degrees of the vertices. For this degree
sequence, we can define a CL model. We refer to this as the \emph{associated CL model}
for a given set of SKG parameters. This CL model will be used to define  a probability matrix $P_{\textrm{CL}}$. In
this paper, we will study the relations between $P_{\textrm{SKG}}$ and $P_{\textrm{CL}}$. Whenever we use
the term $P_{\textrm{CL}}$, this will always be the associated CL model of some SKG matrix $P_{\textrm{SKG}}$.

\subsection{Our Contributions}
\label{sec:contributions}

The main message of this work can be stated simply. The SKG model is close
enough to its associated CL model that most users of SKG could just as well use the CL model
for generating graphs. These models have very similar properties  both in terms of ease of use and in terms of the graphs they generate. Moreover, they both
reflect real data to the same extent. The general CL model has the major advantage of generating
any desired degree distribution.

We stress that we do not claim that the CL model accurately represents real graphs, 
or is even the ``right" model to think about. But we feel that it is a good control model, and it
is one that any other model should be compared against. Fitting
CL to a given graph is quite trivial; simply feed the degree distribution of the real
graph to the CL model. Our results suggest
that users of SKG can satisfy most of their needs with a CL model. 

We provide evidence for this in three different ways.

\begin{asparaenum}
	\item \emph{Graph properties of SKG vs CL:} 
          We construct an SKG using known parameter choices
	from the Graph500 specification.
	We then generate CL graphs with the same degree distributions. The comparison
	of graph properties is very telling. The degree distribution are naturally very similar. What is
	surprising is that the clustering coefficients, eigenvalues, and core decompositions match exceedingly
	well. Note that the CL model can be thought of as a uniform random samples of graphs with
	an input degree distribution. It appears that SKG is very similar, where the degree distribution
	is given implicitly by the generator matrix $T$.

	\item \emph{Quantitative comparison of generating matrices $P_{\textrm{SKG}}$ and $P_{\textrm{CL}}$:} 
          We propose an explanation of these observations
	based on comparisons of the probability matrices of SKG
	and CL. We plot the entries of these matrices in various ways,
	and arrive at the conclusion that these matrices are extremely similar. More concretely,
	they represent almost the same distribution on graphs, and differences are very slight.
	This strongly suggests that the CL model is a good
	and simple approximation of SKG, and it has the additional benefit of modeling any degree distribution.
	We prove that under a simple condition on the matrix $T$, 
	$P_{\textrm{SKG}}$ is \emph{identical} to $P_{\textrm{CL}}$. Although this condition is
	often not satisfied by common SKG parameters, it gives strong mathematical intuition behind the 
	similarities. 

	\item \emph{Comparing SKG and CL to real data:} The popularity of SKG is significantly due to fitting
	procedures that compute SKG parameters corresponding to real graphs \cite{LeChKlFa10}.
	This is based on an expensive likelihood optimization procedure. Contrast this with CL, which
	has a trivial fitting mechanism. We show that both these models do a similar job
	of matching graph parameters. Indeed, CL guarantees to fit the degree distribution (up to
	expectations). In other graph properties, neither SKG nor CL is clearly better.
	This is a very compelling reason to consider the CL model as a control model.
\end{asparaenum}

\medskip
In this paper, we focus primarily on SKG instead of the noisy version NSKG because SKG is extremely well
established and used by a large number of researchers \cite{Vineet:2009, 5713180, Pearce:2010, BulucGilbert, Pothen, Schmidt2009417, HOO11, Plimpton2011610, 4536261, 4454426}.
Nonetheless, all our experiments and comparisons are also performed with NSKG. Other than
correcting deficiencies in the degree distribution, the effect of noise
on other graph properties seems fairly small. Hence, for our matrix studies and 
mathematical theorems (\Sec{matrix} and \Sec{math}), we focus on similarities between SKG and CL. We however note
that all our empirical evidence holds for NSKG as well: CL seems to model NSKG graphs
reasonably well (though not as perfectly as SKG), and CL fits real data as well as
NSKG.

\begin{figure*}[htb]

  \centering
\hspace*{-5ex}
  \subfloat[Degree distribution]{\label{fig:g500-deg}
  \includegraphics[width=.7\columnwidth,trim=0 0 0 0]{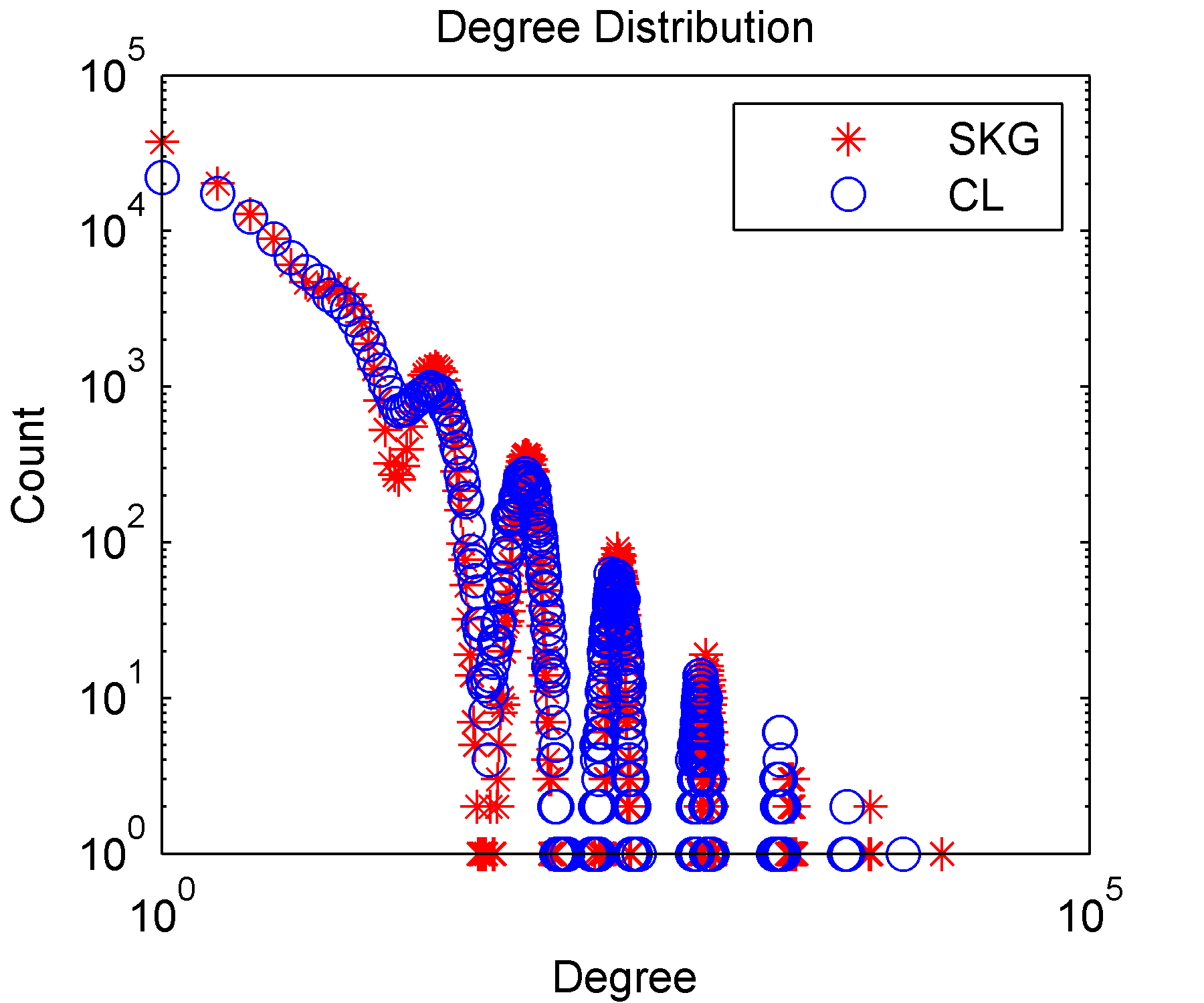}
  }
  \subfloat[Clustering coefficients]{\label{fig:g500-cc}
  \includegraphics[width=.7\columnwidth,trim=0 0 0 0]{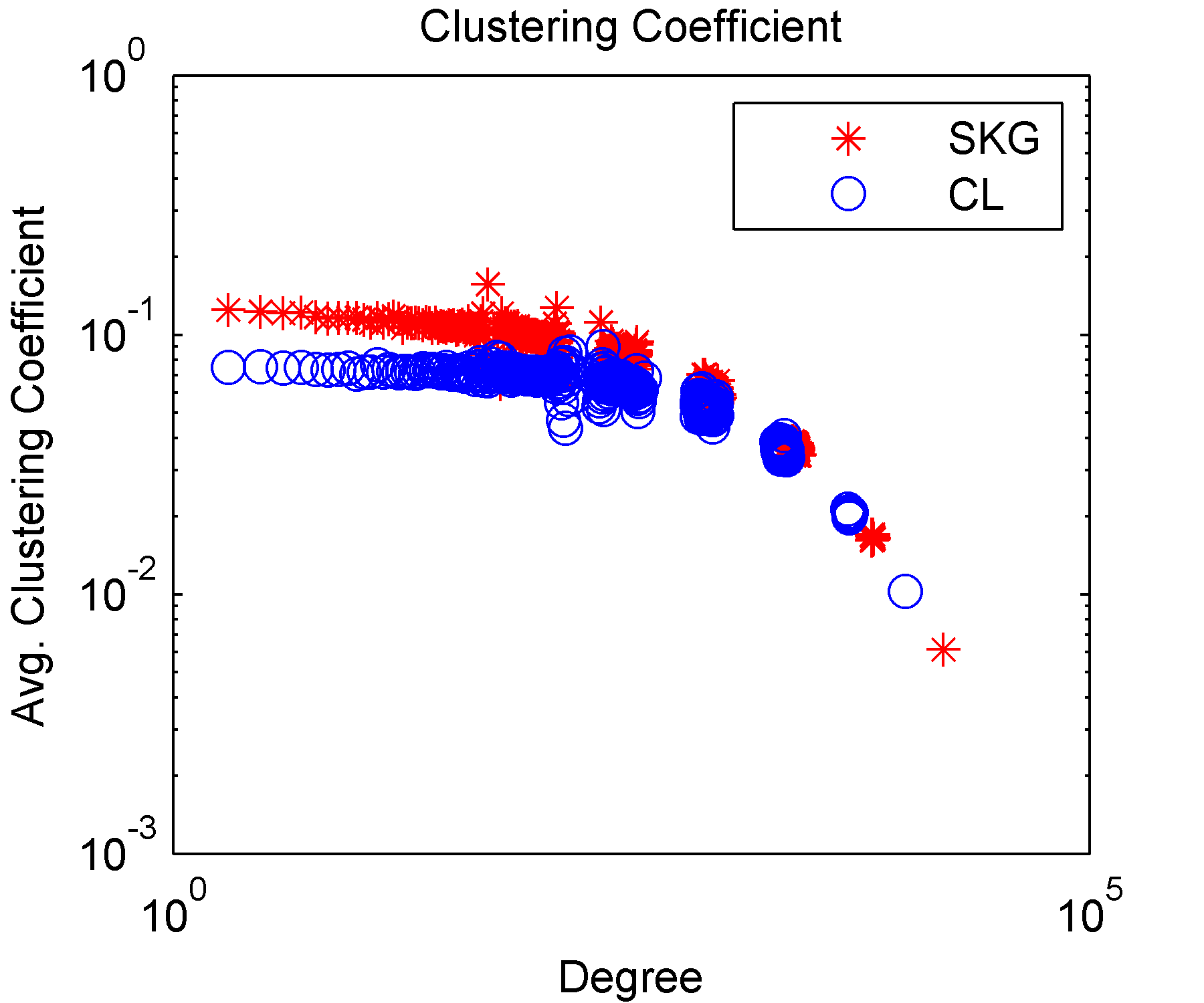}
  }
  \subfloat[Eigenvalues of adjacency matrices]{\label{fig:g500-eig}
  \includegraphics[width=.7\columnwidth,trim=0 0 0 0]{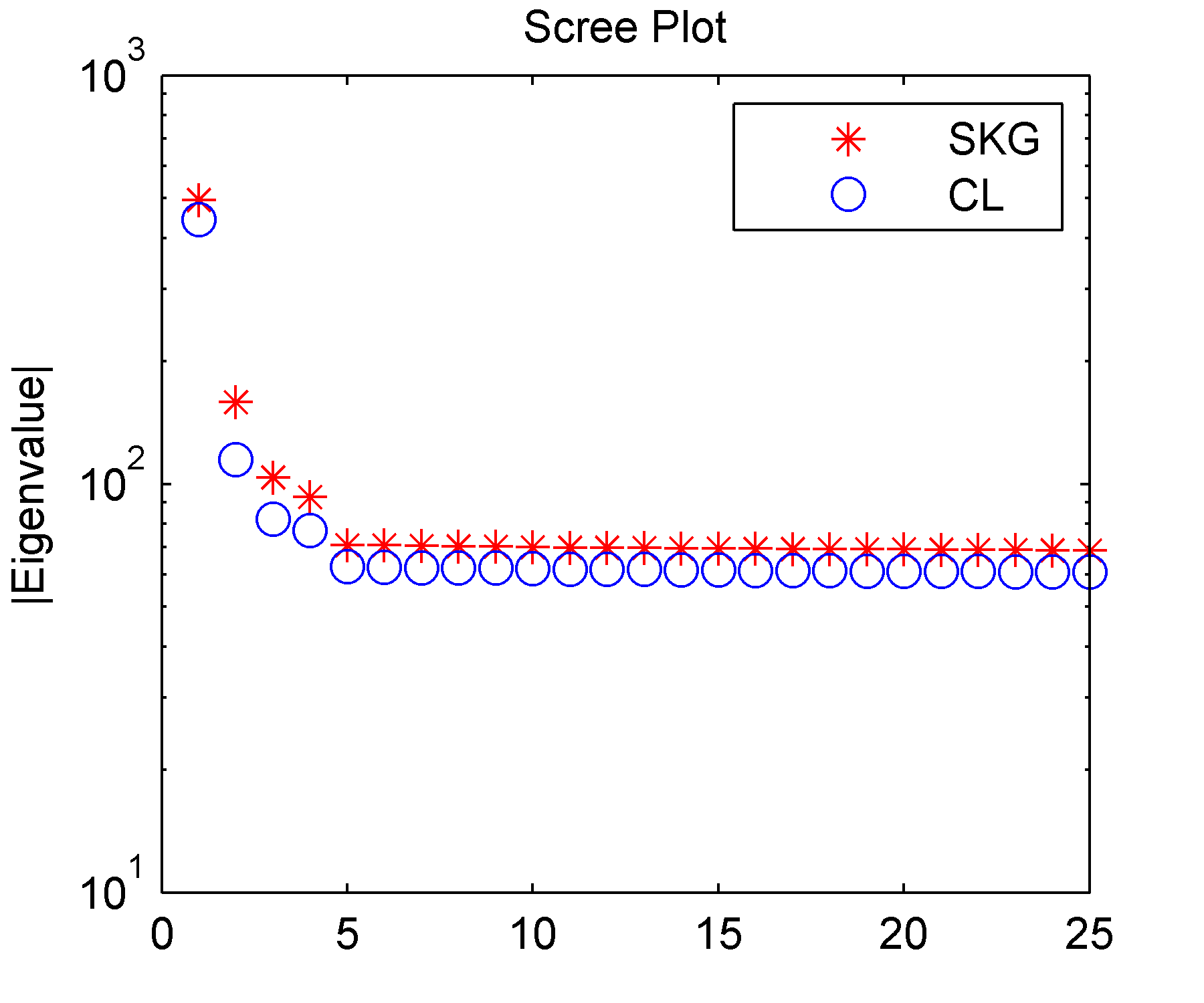}
  }\\
  \subfloat[Assortativity]{\label{fig:g500-assort}
  \includegraphics[width=.7\columnwidth,trim=0 0 0 0]{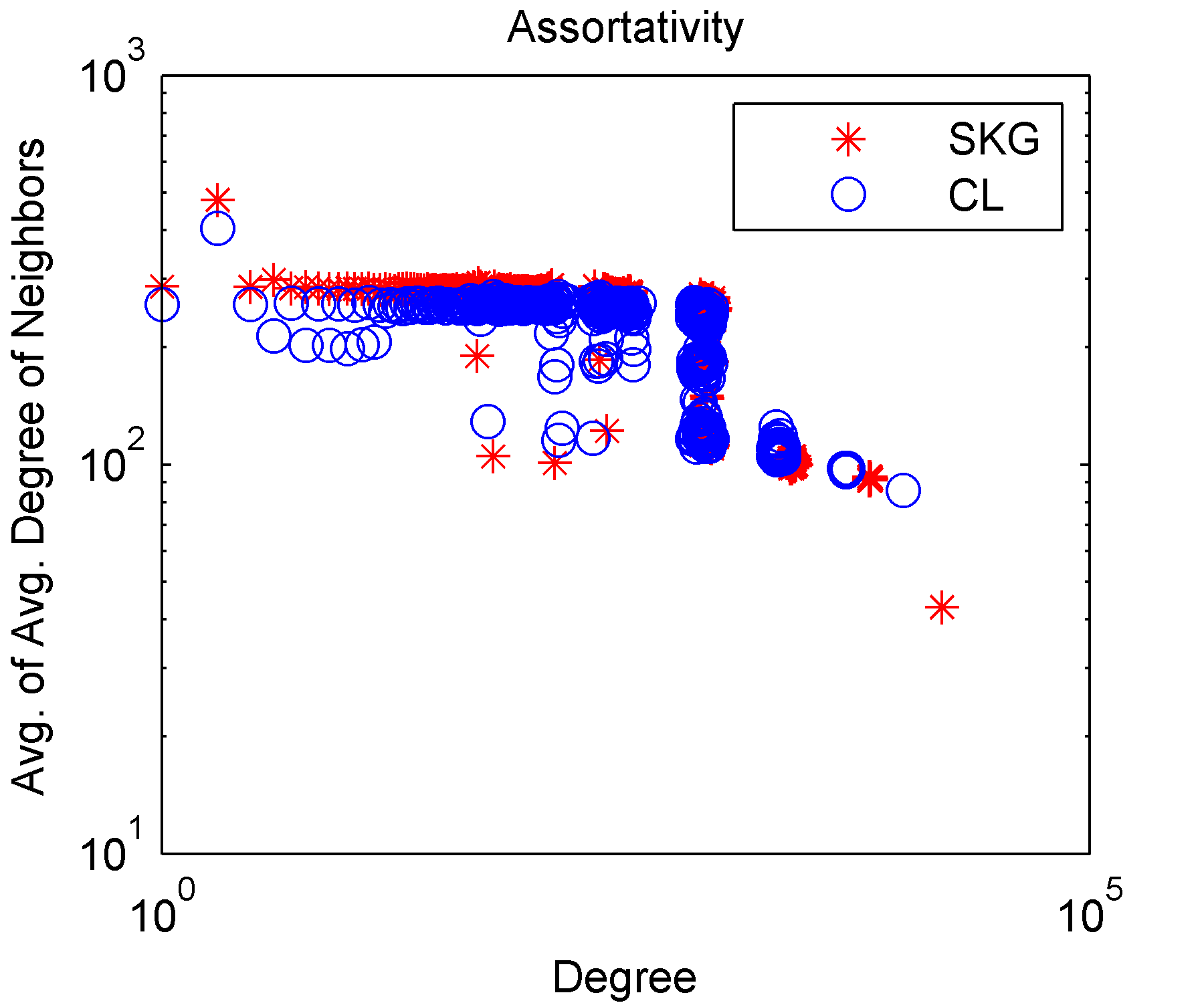}
  }
  \subfloat[Core decompositions]{\label{fig:g500-core}
  \includegraphics[width=.7\columnwidth,trim=0 0 0 0]{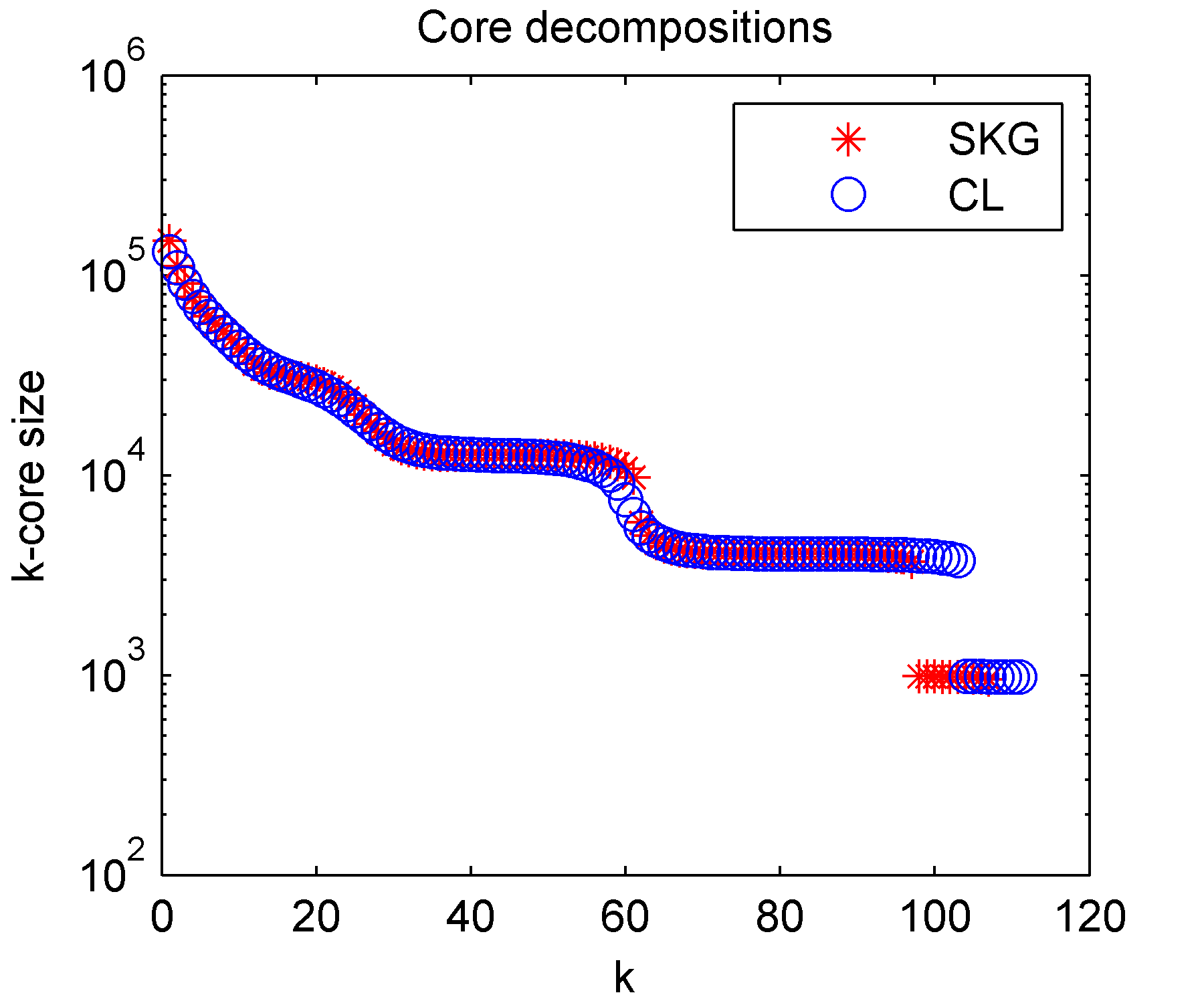}
  }  
  \caption{
  Comparison of  the graph properties of SKG generated with Graph500 parameters and
  an equivalent CL.}
  \label{fig:g500}
\end{figure*}

\subsection{Parameters for empirical study}

We focus attention on the Graph500 benchmark \cite{Graph500}. 
This is primarily for concreteness and the relative importance
of this parameter setting. Our results hold for all the settings of parameters that
we experimented with. For NSKG, there is an additional noise parameter required.
We set this to $0.1$, the setting studied in \cite{SePiKo11}.

\begin{asparaitem}
	\item \textbf{Graph500}: $T=[0.57, 0.19; 0.19, 0.05]$, $\ell \in \{$26, 29, 32, 36, 39, 42$\}$, and $m = 16 \cdot 2^\ell$. We focus on a much smaller setting, $\ell = 18$. 
\end{asparaitem}

\section{Previous Work}
\label{sec:related}

The SKG model was proposed by Leskovec et al.\@ \cite{LeChKlFa05}, as a generalization
of the R-MAT model, given by Chakrabarti et al.\@ \cite{ChZhFa04}. Algorithms
to fit SKG to real data were given by Leskovec and Faloutsos \cite{LeFa07} (extended in \cite{LeChKlFa10}).
This model has been chosen for the 
Graph500 benchmark \cite{Graph500}. 
Kim and Leskovec~\cite{KiLe10} defined a variant of SKG called the Multiplicative Attribute Graph (MAG) model. 

There have been various analyses of the SKG model.
The original paper \cite{LeChKlFa10} provides some basic theorems and empirically shows
a variety of properties. Gro\"er et al.\@ \cite{GrSuPo10}, Mahdian and Xu~\cite{MaXu10},
and Seshadhri et al.\@ \cite{SePiKo11} study how
the model parameters affect the graph properties. It has been conclusively shown
that SKG cannot generate power-law distributions \cite{SePiKo11}. Seshadhri et al.\@ also proposed  noisy SKG (NSKG), which  can provably produce lognormal degree distributions.

Sala et al.\@ \cite{SaCaWiZa10} perform an extensive
empirical study of properties of graph models, including SKGs.
Miller et al.\@ \cite{MiBlWo10} give algorithms to detect anomalies embedded in an SKG.
Moreno et al.\@ \cite{MoKiNeVi10} study the distributional properties of families of SKGs.

A good survey of the edge-configuration model and its variants is given
by Newman \cite{Ne03} (refer to Section IV.B). The specific model of CL was first
given by Chung and Lu \cite{ChLu02, ChLu02-2}. They proved many properties of these 
graphs. Properties of its eigenvalues were given by Mihail and Papadimitriou \cite{MiPa02} and
Chung et al.\@ \cite{ChLuVu03}.

\begin{figure*}[htb]

  \centering
\hspace*{-5ex}
  \subfloat[Degree distribution]{\label{fig:n-g500-deg}
  \includegraphics[width=.7\columnwidth,trim=0 0 0 0]{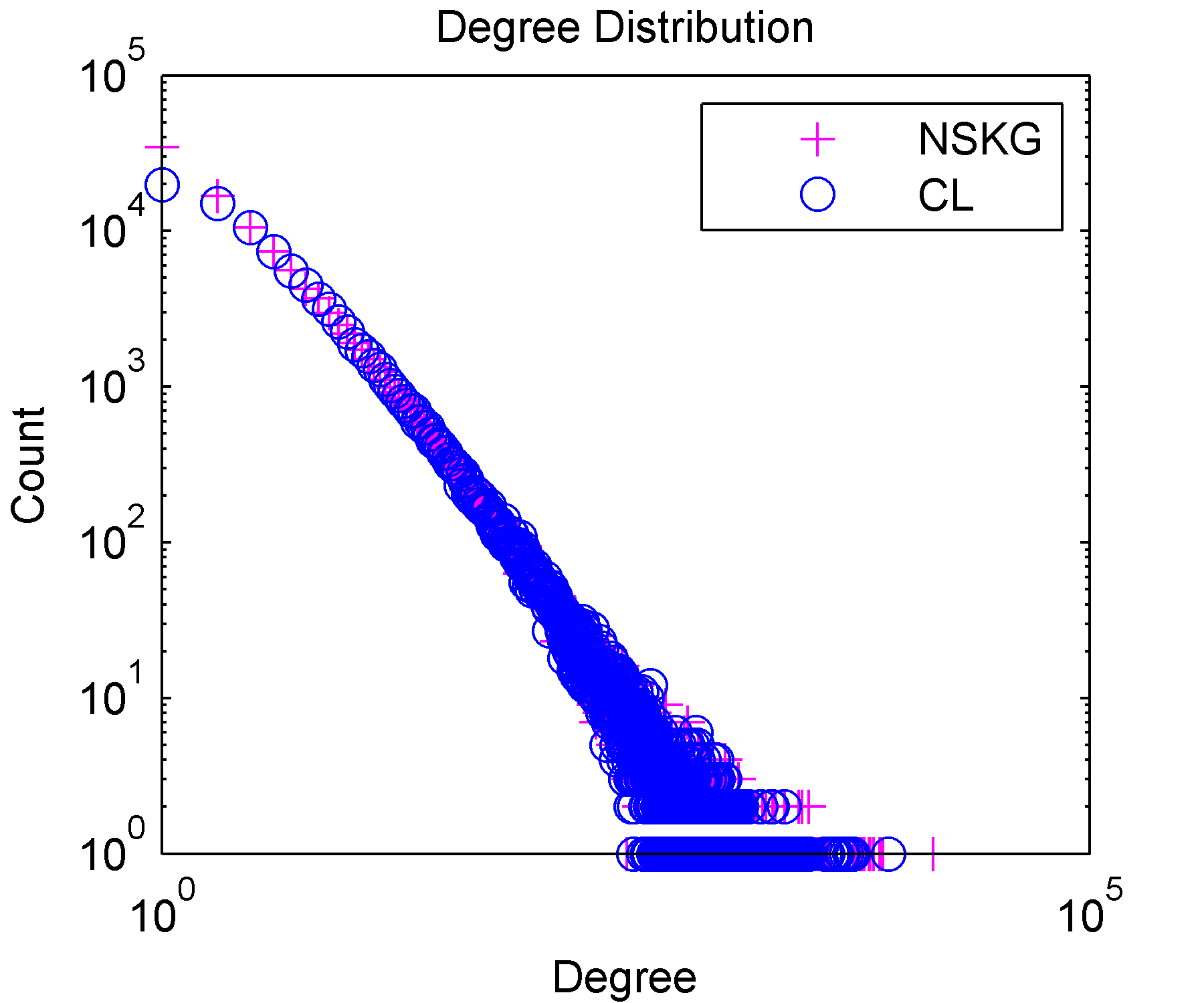}
  }
  \subfloat[Clustering coefficients]{\label{fig:n-g500-cc}
  \includegraphics[width=.7\columnwidth,trim=0 0 0 0]{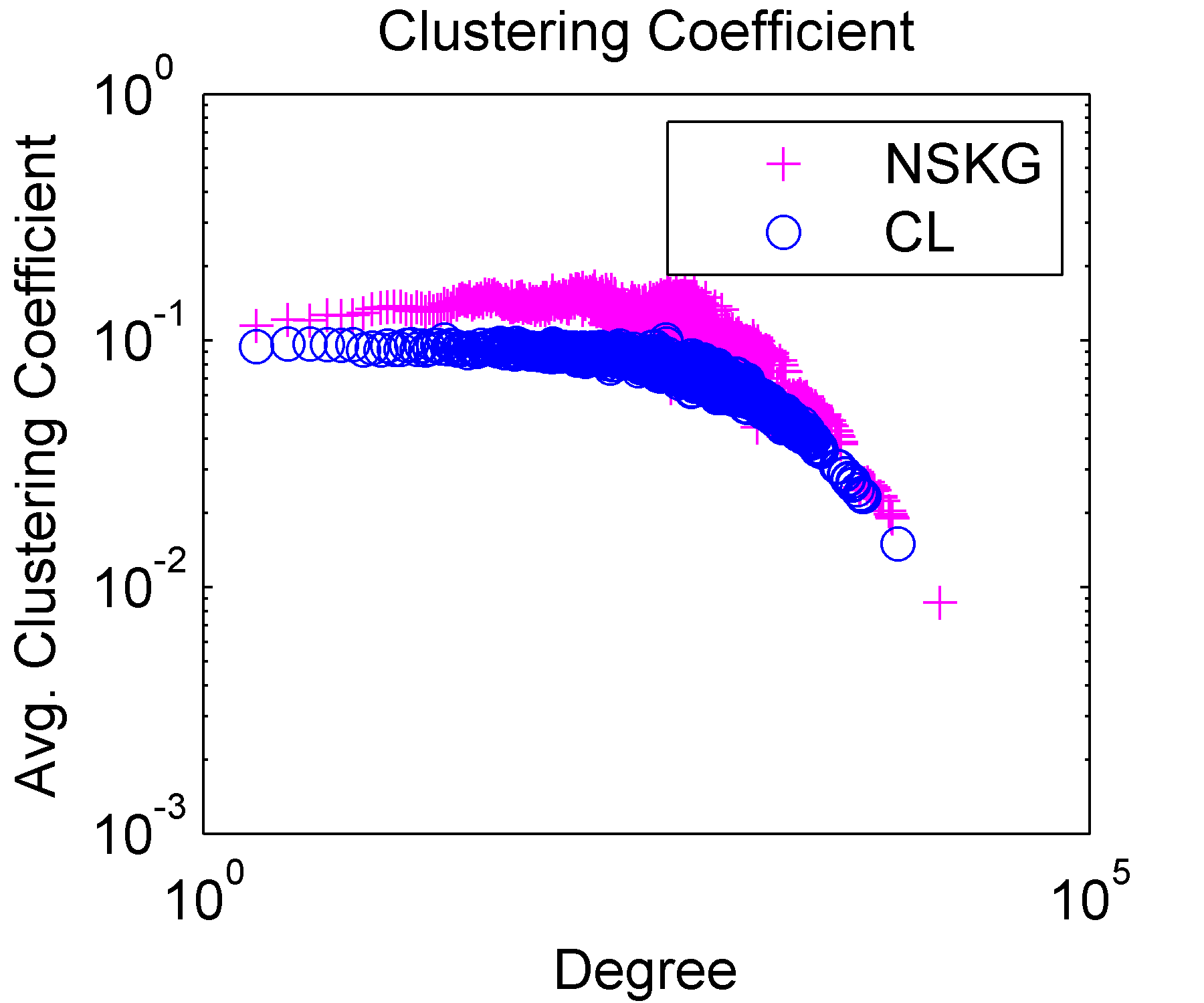}
  }
  \subfloat[Eigenvalues of adjacency matrices]{\label{fig:n-g500-eig}
  \includegraphics[width=.7\columnwidth,trim=0 0 0 0]{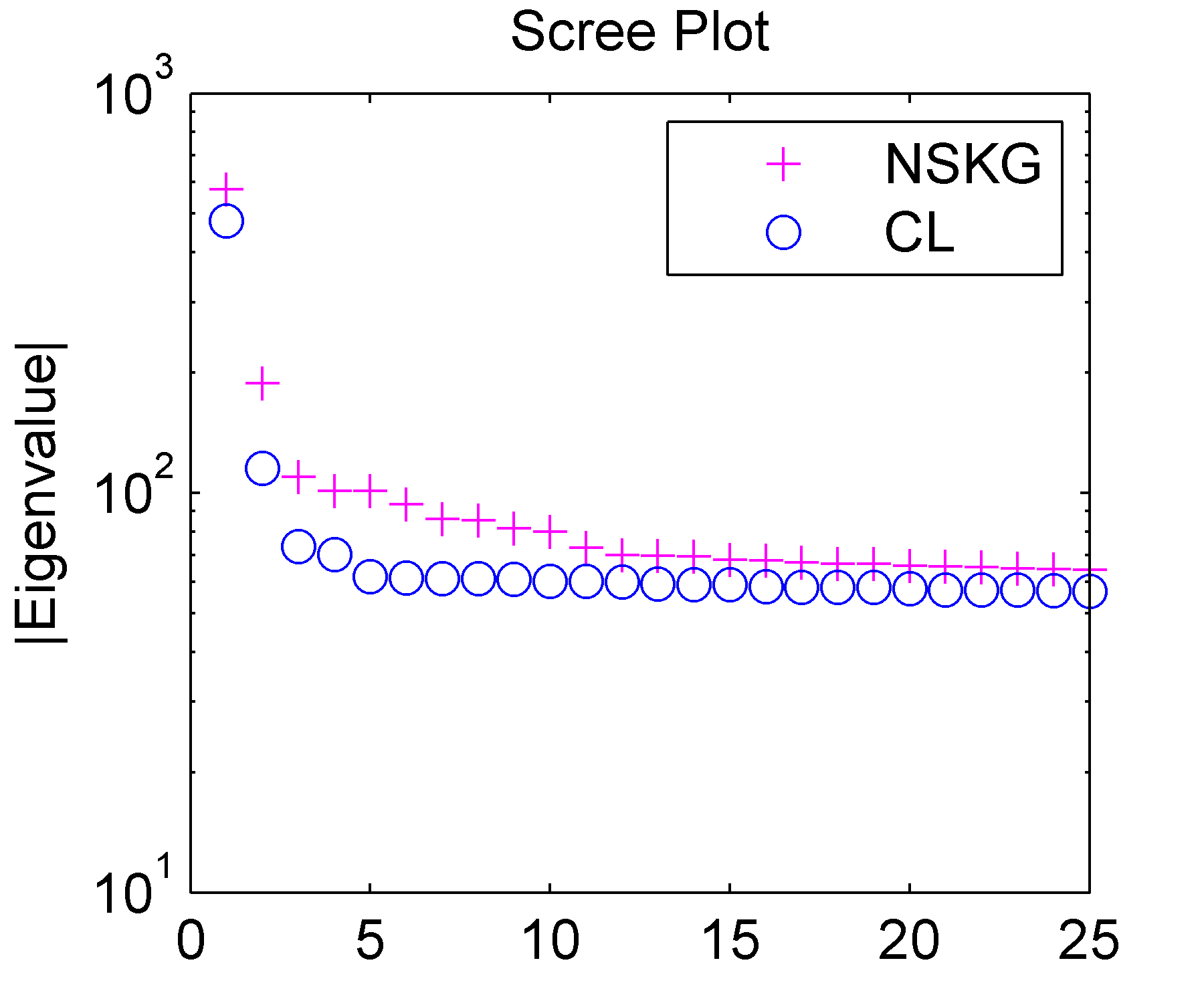}
  }\\
  \subfloat[Assortativity]{\label{fig:n-g500-assort}
  \includegraphics[width=.7\columnwidth,trim=0 0 0 0]{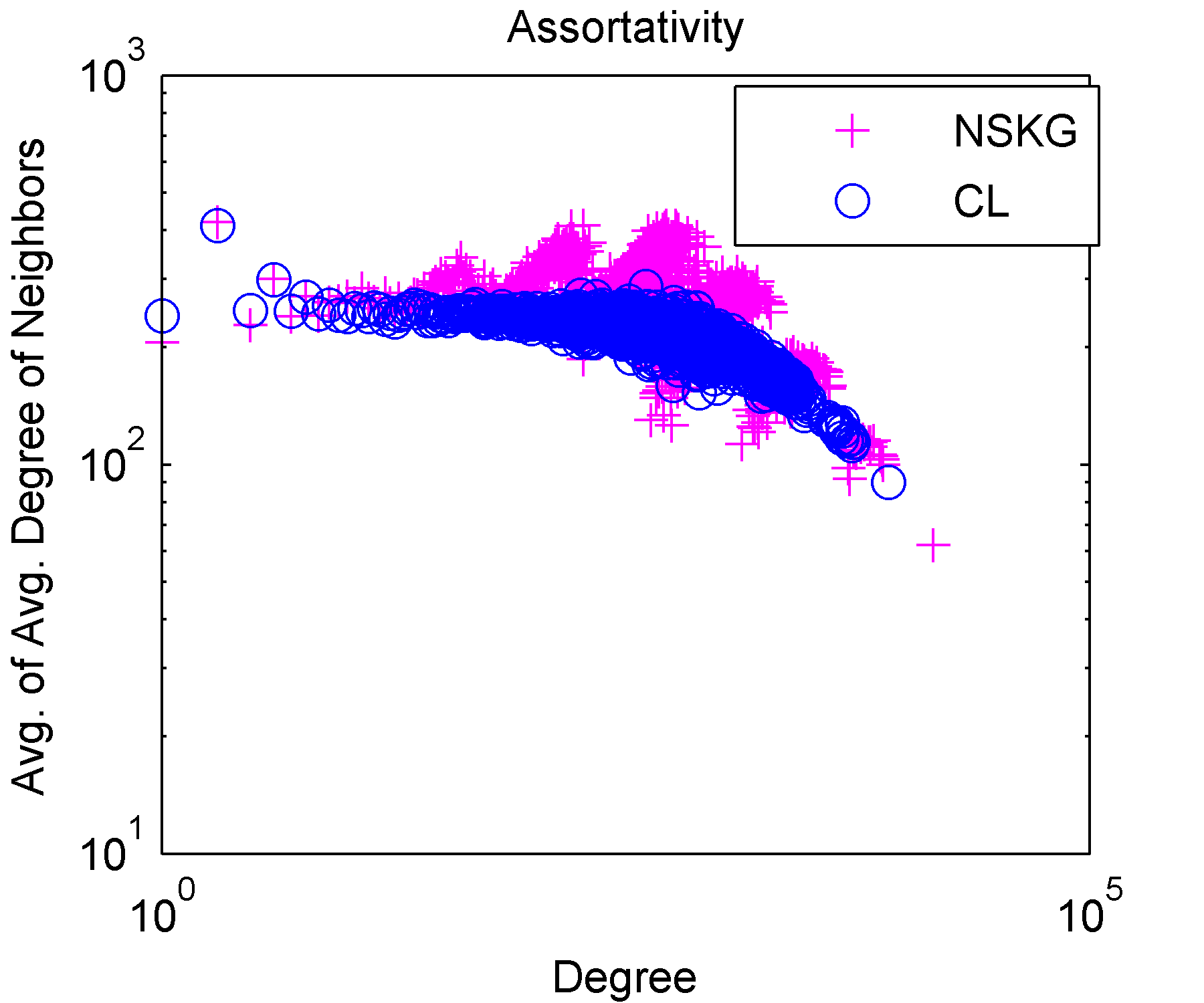}
  }
  \subfloat[Core decompositions]{\label{fig:n-g500-core}
  \includegraphics[width=.7\columnwidth,trim=0 0 0 0]{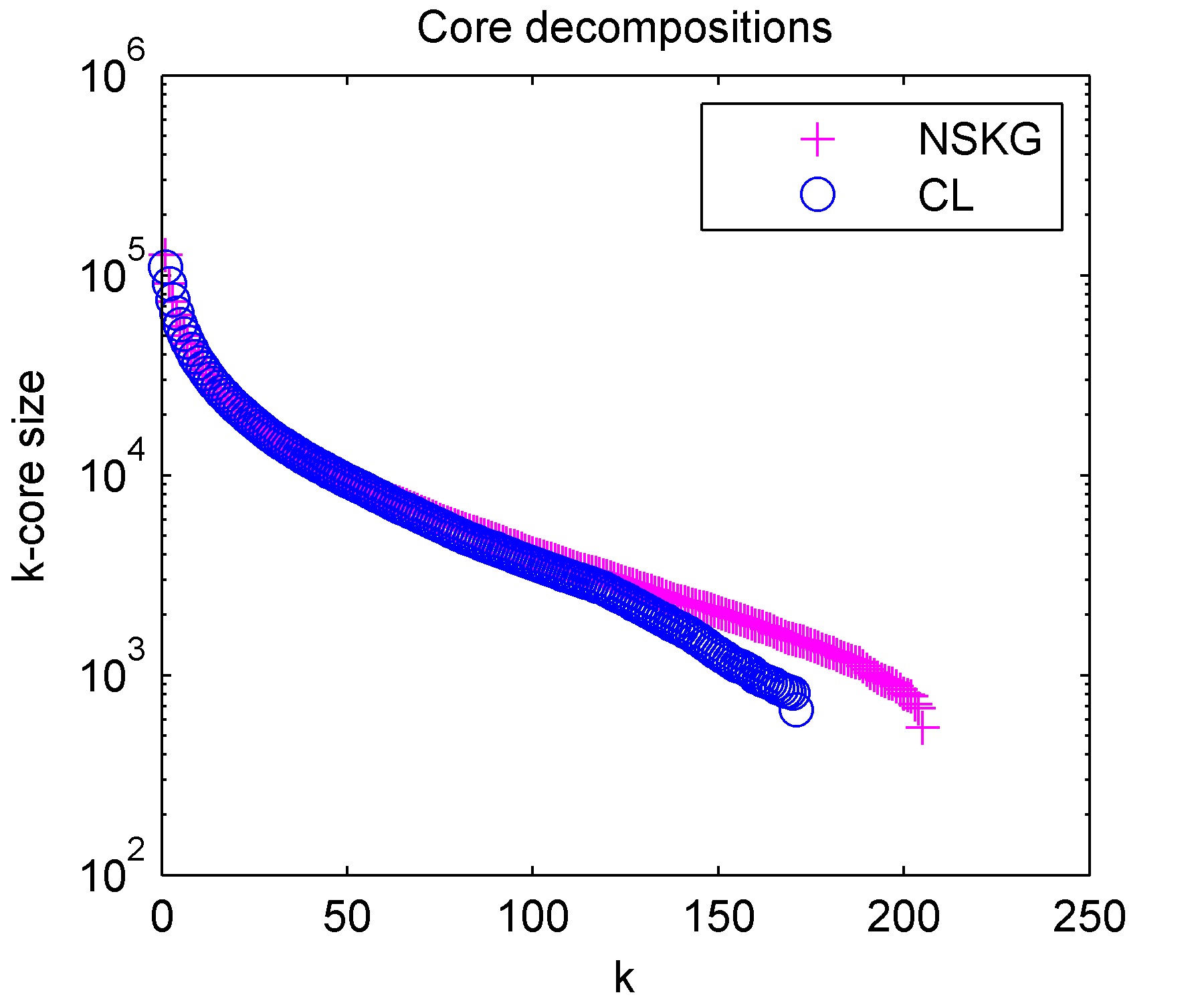}
  }  
  \caption{
  The figures compares the graph properties of NSKG generated with Graph500 parameters and
  an equivalent CL.}
  \label{fig:n-g500}
\end{figure*}
\section{Similarity between SKG and CL} \label{sec:simi}

Our first experiment details the similarities between an SKG and its equivalent CL. We construct
an SKG using the Graph500 parameters with $\ell = 18$. We take the degree distribution
of this graph, and construct a CL graph using this. Various properties
of these graphs are given in \Fig{g500}. We give details below:

\begin{asparaenum}
	\item {\it Degree distribution} (\Fig{g500-deg}): This is the standard degree distribution plot in log-log
	scale. It is no surprise that the degree distributions are almost identical. After all, the weighting
	of CL is done precisely to match this.
	\item {\it Clustering coefficients} (\Fig{g500-cc}): The clustering coefficient of a vertex $i$ is the fraction
	of wedges centered at $i$ that participate in triangles. We plot $d$ versus the average clustering
	coefficient of a degree $d$ vertex in log-log scale. Observe the close similarity. Indeed, we measure 
	the difference between clustering coefficient values at $d$ to be at most $0.04$ (a lower order term
	with respect to commonly measured values in real graphs \cite{GiNe02}).
	\item {\it Eigenvalues} (\Fig{g500-eig}): Here, we plot the first 25 eigenvalues (in absolute value) of the adjacency matrix of the graph   
	in log-scale. The proximity of eigenvalues is very striking. This is a strong suggestion that graph
	structure of the SKG and CL graphs are very similar.
	\item {\it Assortativity} (\Fig{g500-assort}): This is non-standard measure, but we feel that it provides
	a lot of structural intuition. Social networks are often seen to be assortative \cite{New02-1,New02-2}, which means
	that vertex of similar degree tend to be connected by edges. For $d$, define $X_d$ to be the average degree
	of an average degree $d$ vertex. We plot $d$ versus $X_d$ in log-log scale. Note that neither SKG nor CL are
	particularly assortative, and the plots match rather well.
	\item {\it Core decompositions} (\Fig{g500-core}): The $k$-cores of a graph are a very important part of understanding
	the community structure of a graph. The size of the $k$-core is the largest induced subgraph where each vertex has a minimum degree of $k$. This is a subset $S$
	of vertices such that \emph{all} vertices have $k$ neighbors in $S$. These sizes can be quickly determined by performing a \emph{core
	decomposition}. This is obtained by iteratively deleting the minimum degree vertex of the graph. The core plots look 
	amazingly close, and the only difference is that there are slightly larger cores in CL.
\end{asparaenum}
\medskip

All these plots clearly suggest that the Graph500 SKG and its equivalent CL graph are incredibly close in their
graph properties. Indeed, it appears that most important structural properties (especially from a social networks
perspective) are closely related. We will show in \Sec{fit} that CL performs an adequate job of fitting
real data, and is quite comparable to SKG.
We feel that any uses of SKG for benchmarking or test instances generation could
probably be done with CL graphs as well.

For completeness, we plot the same comparisons between NSKG and CL in \Fig{n-g500}. We note
again that the properties are very similar, though NSKG shows more variance in its values.
Clustering coefficient values differ by at most $0.02$ here, and barring small differences
in initial eigenvalues, there is a very close match. The assortativity plots show more oscillations
for NSKG, but CL gets the overall trajectory.

\section{Connection between SKG and CL matrices} \label{sec:matrix}

Is there a principled explanation for the similarity observed in \Fig{g500}? It appears to be much more
than a coincidence, considering the wide variety of graph properties that match. In this section,
we provide an explanation based on the similarity of the probability matrices $P_{\textrm{SKG}}$ and $P_{\textrm{CL}}$. On analyzing
these matrices, we see that they have an extremely close distribution of values. These
matrices are themselves so fundamentally similar, providing more evidence that SKG
itself can be modeled as CL.

We begin by giving precise formulae for the entries of the SKG and CL matrices. This is by no means
new (or even difficult), but it should introduce the reader to the structure of these matrices.
The vertices of the graph are labeled from $[n]$ (the set of positive integers up to $n$). 
For any $i$, $\bv_i$ denotes the 
binary representation of $i$ as an $\ell$-bit vector. For two vectors $\bv_i$ and $\bv_j$,
the number of common zeroes is the number of positions where both vectors are $0$.
The following formula for the SKG entries has already been used in \cite{ChZhFa04,LeChKlFa05,GrSuPo10}.
Observe that these entries (for both SKG and CL) are quite easy to compute and enumerate.

\medskip

\begin{claim} \label{clm:entries} Let $i,j \in [n]$. Let the number of zeroes in $\bv_i$ and $\bv_j$
be $z_i$ and $z_j$ respectively. Let the number of \emph{common} zeroes
be $c_z$. Then
\begin{align*}
&  P_{\textrm{SKG}}(i,j)  = t_1^{c_z}t_2^{z_i - c_z}t_3^{z_j - c_z}t_4^{\ell - z_i - z_j + c_z}, \text{ and} \\
 & P_{\textrm{CL}}(i,j)\!  = \!(t_1+t_2)^{z_i}(t_3+t_4)^{\ell-z_i}(t_1+t_3)^{z_j}(t_2+t_4)^{\ell-z_j}.
\end{align*}
\end{claim}

\begin{proof} The number of positions where $\bv_i$ is zero but 
$\bv_j$ is one is $z_i - c_z$. Analogously, the number of positions
where only $\bv_j$ is zero is $z_j - c_z$. The number of common ones
is $\ell - z_i - z_j + c_z$.
Hence, the $(i,j)$ entry of the $P_{\textrm{SKG}}$ is 
$t_1^{c_z}t_2^{z_i - c_z}t_3^{z_j - c_z}t_4^{\ell - z_i - z_j + c_z}$.

Let us now compute the $(i,j)$ entry in the corresponding CL
matrix. The probability that a single edge insertion becomes an out-edge of vertex $i$
in SKG is $(t_1+t_2)^{z_i}(t_3+t_4)^{\ell-z_i}$. Hence, the expected out-degree
of $i$ is $m(t_1+t_2)^{z_i}(t_3+t_4)^{\ell-z_i}$. 
Similarly, the expected in-degree of $j$ is  $m(t_1+t_3)^{z_j}(t_2+t_4)^{\ell-z_j}$.
The $(i,j)$ entry of $P_{\textrm{CL}}$ is
$(t_1+t_2)^{z_i}(t_3+t_4)^{\ell-z_i}(t_1+t_3)^{z_j}(t_2+t_4)^{\ell-z_j}.\myproofend$
\end{proof}

\medskip

The inspiration for this section comes from \Fig{skgentries}. Our initial aim was to understand the SKG matrix, and see whether
the structure of the values provides insight into the properties of SKG. Since each entry in this probability matrix
is of the form $t_1^{c_z}t_2^{z_i - c_z}t_3^{z_j - c_z}t_4^{\ell - z_i - z_j + c_z}$, there are many repeated values in
this matrix. For each value in this probability matrix, we simply plot the number of times (the multiplicity) this value
appears in the matrix. (For $P_{\textrm{SKG}}$, this is given in red) This is done for the associated
$P_{\textrm{CL}}$ in blue. Note the uncanny similarity
of the overall shapes for SKG and CL. 
Clearly, $P_{\textrm{SKG}}$ has more distinct values\footnote{This can be proven by inspecting
\Clm{entries}.}, but they are distributed fairly similarly
to $P_{\textrm{CL}}$. Nonetheless, this picture is not very formally convincing, since it only shows the overall behavior of the distribution
of values.

\begin{figure}[thb]
\centerline{
\includegraphics[width=0.5\textwidth]{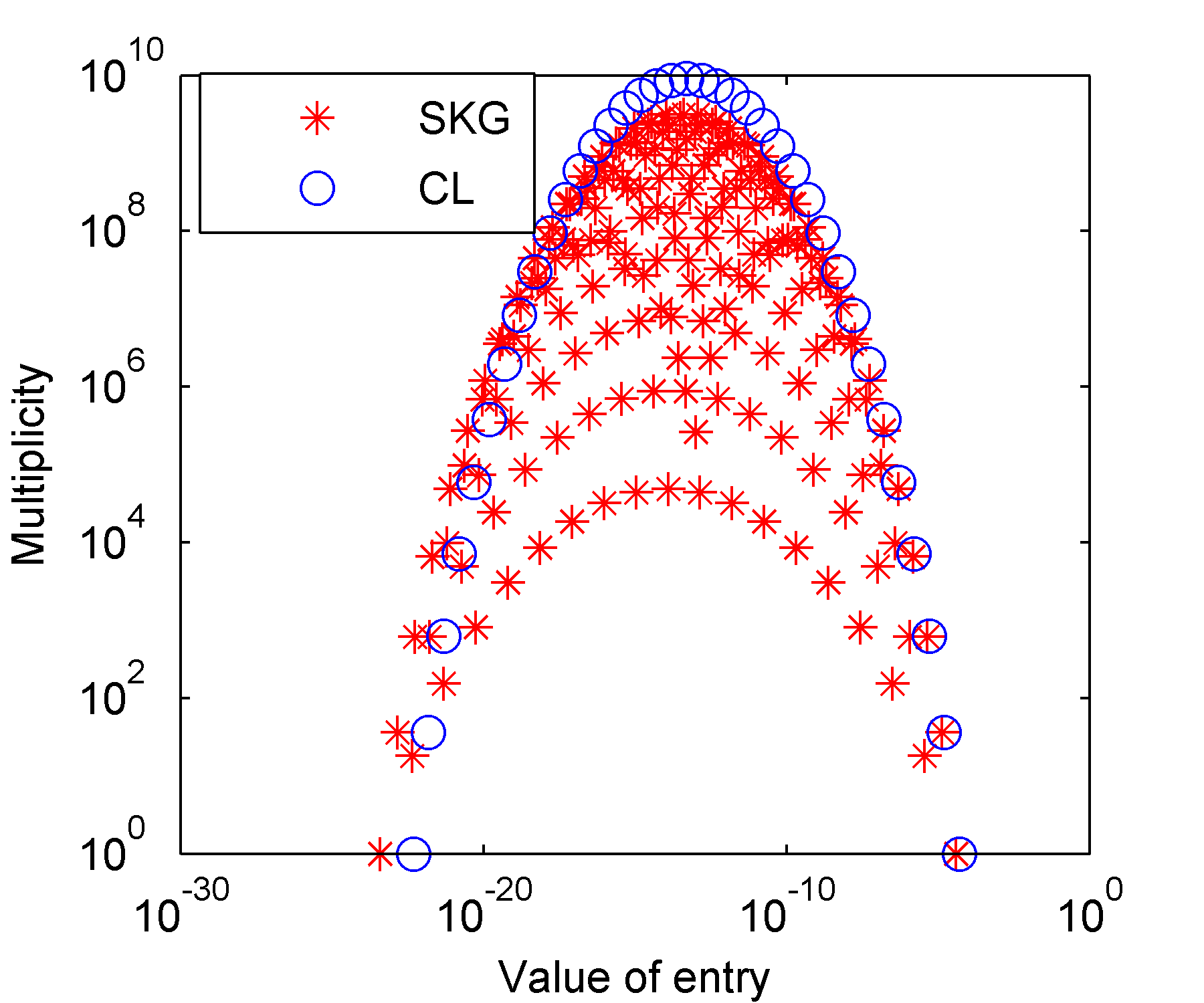} }
\caption{ \label{fig:skgentries}  Distribution of entries of $P_{\textrm{SKG}}$ and $P_{\textrm{CL}}$. }
\end{figure}

\Fig{bin} makes a more faithful comparison between the $P_{\textrm{SKG}}$ and $P_{\textrm{CL}}$ matrices. As we note from \Fig{skgentries}, $P_{\textrm{CL}}$ has a much smaller set of distinct entries. Suppose the distinct values
in $P_{\textrm{CL}}$ are $v_1 > v_2 > v_3 \ldots$. Associate a bin with each distinct entry of $P_{\textrm{CL}}$. 
For each entry of $P_{\textrm{SKG}}$,
place it in the bin corresponding to the entry of $P_{\textrm{CL}}$ with the \emph{closest value}.
So, if some entry in $P_{\textrm{SKG}}$ has value $v$, we determine the index $i$ such that
$|v - v_i|$ is minimized. This entry is placed in the $i$th bin. We can now look at the size
of each bin for $P_{\textrm{SKG}}$. The size of the $i$th bin for the $P_{\textrm{CL}}$ is simply
the multiplicity of $v_i$ in $P_{\textrm{CL}}$.

Observe how these sizes are practically \emph{identical} for large enough entry value.
Indeed the former portion of these plots, for value $< 10^{-20}$, only accounts
for a total of $< 10^{-5}$ of the probability mass. This means that the fraction
of edges that will correspond to these entries is at most $10^{-5}$. We can also argue
that these entries correspond only to edges joining very low degree vertices to each other.
In other words, the portion where these curves differ is really immaterial to the structure
of the final graph generated.

\begin{figure}[thb]
\centerline{
\includegraphics[width=0.5\textwidth]{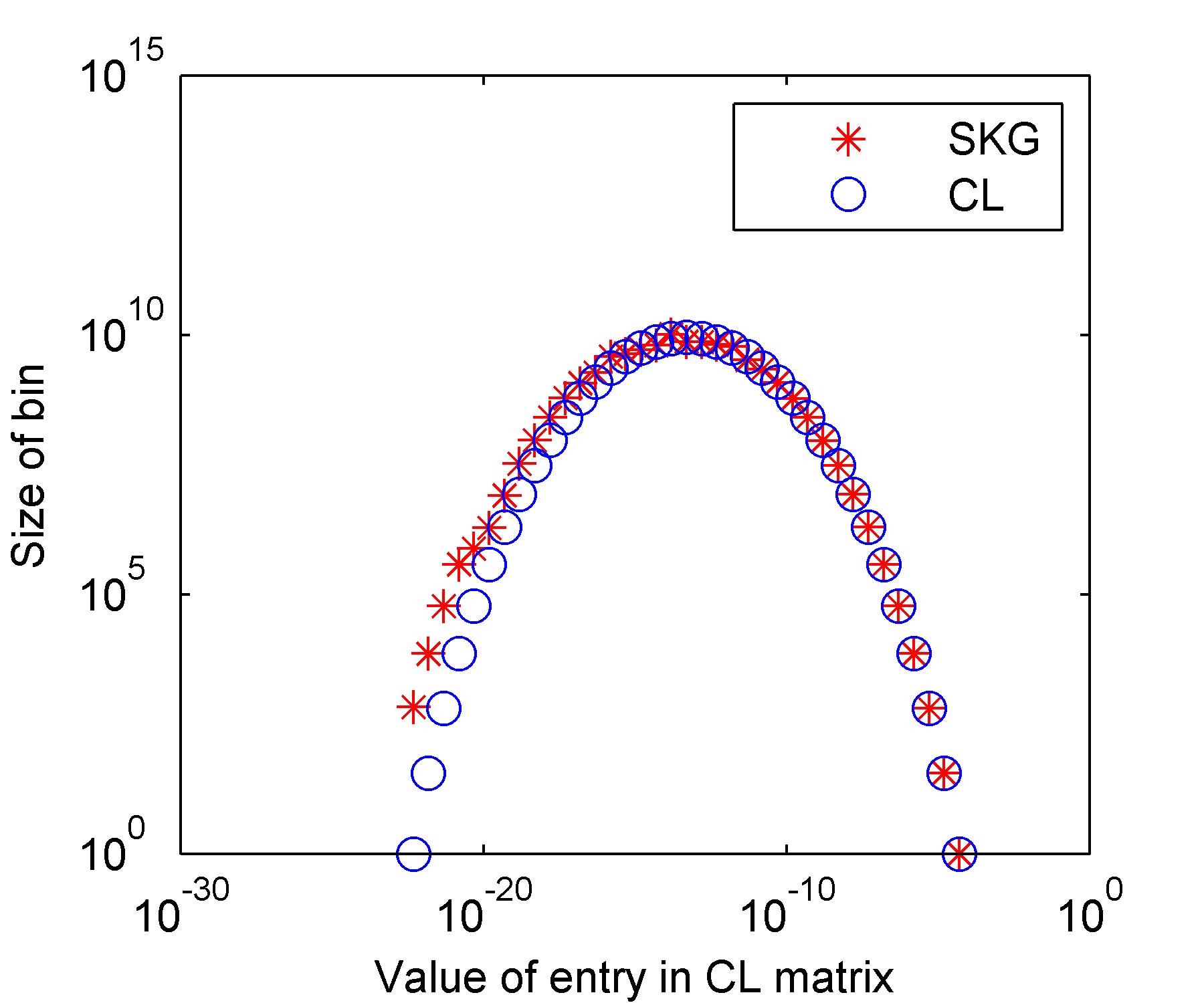} }
\caption{Bin sizes for $P_{\textrm{SKG}}$ and $P_{\textrm{CL}}$ \label{fig:bin}  }
\end{figure} 

This is very strong evidence that SKG behaves like a CL model. The structure of the matrices
are extremely similar to each other. \Fig{cp} is even more convincing. Now, instead of just
looking at the size of each bin, we look at the total probability mass of each bin. For $P_{\textrm{SKG}}$
matrix, this is the sum of entries in a particular bin. For $P_{\textrm{CL}}$, this
is the product of the size of the bin and the value (which is the again just the sum
of entries in that bin). Again, we note the almost exact coincidence of these plots in the
regime where the probabilities matter. Not only are the number of entries in each bin
(roughly) the same, so is the total probability mass in the bin.

\begin{figure}[thb]
\centerline{
\includegraphics[width=0.5\textwidth]{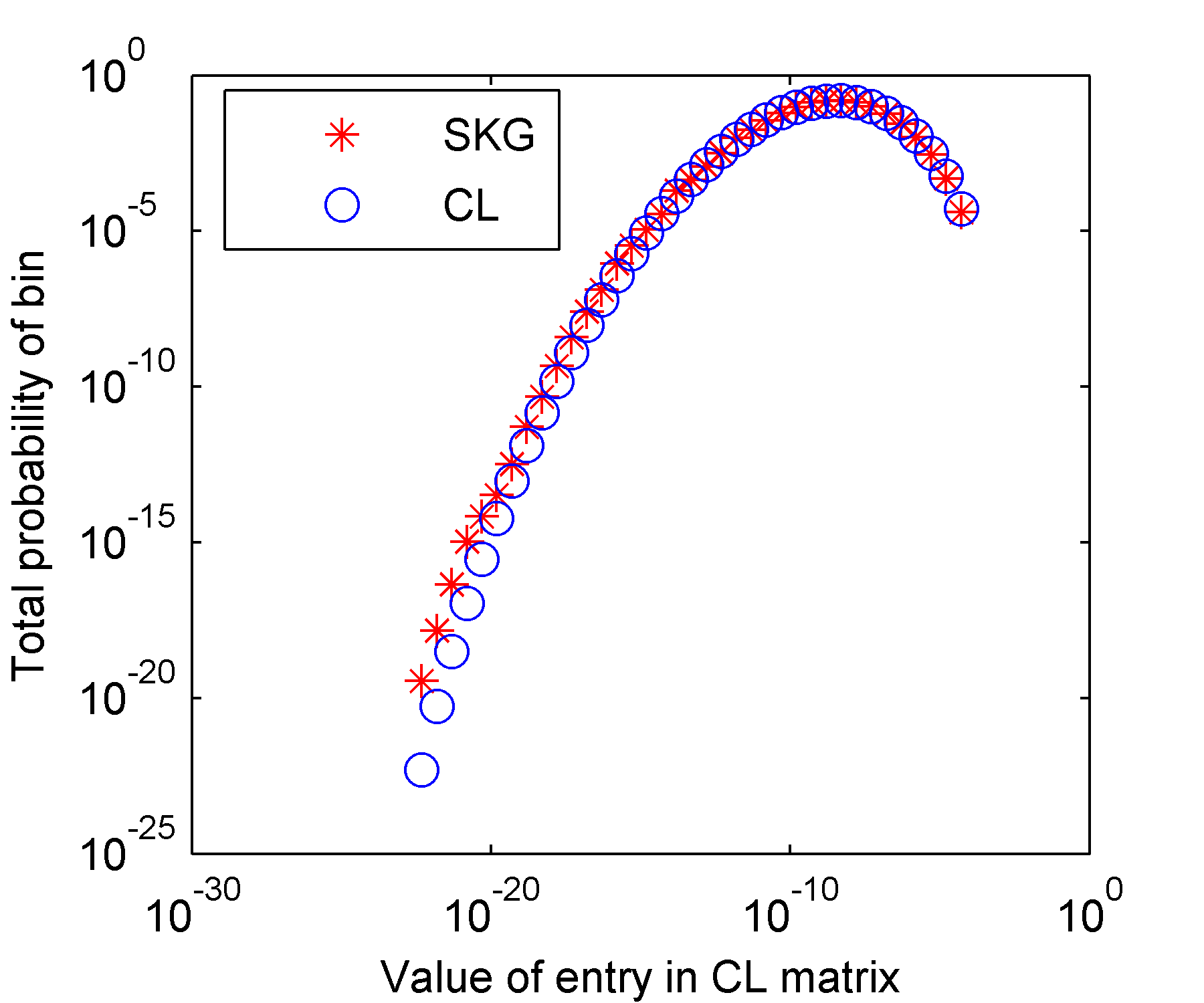} }
\caption{Probability mass of the bins \label{fig:cp}  }
\end{figure} 

We now generate a random sample from $P_{\textrm{SKG}}$ and one from $P_{\textrm{CL}}$. \Fig{spy} shows MATLAB spy plots of the corresponding graphs (represented by their adjacency matrices). 
One of the motivations for the SKG model was that it had a fractal or self-similar structure. It appears
that the CL graph shares the same self-similarity. Furthermore, this self-repetition looks
identical for the both SKG and CL graphs.

\begin{figure*}[t]
  \centering
  \includegraphics[width=\textwidth,trim=0 115 0 110,clip]{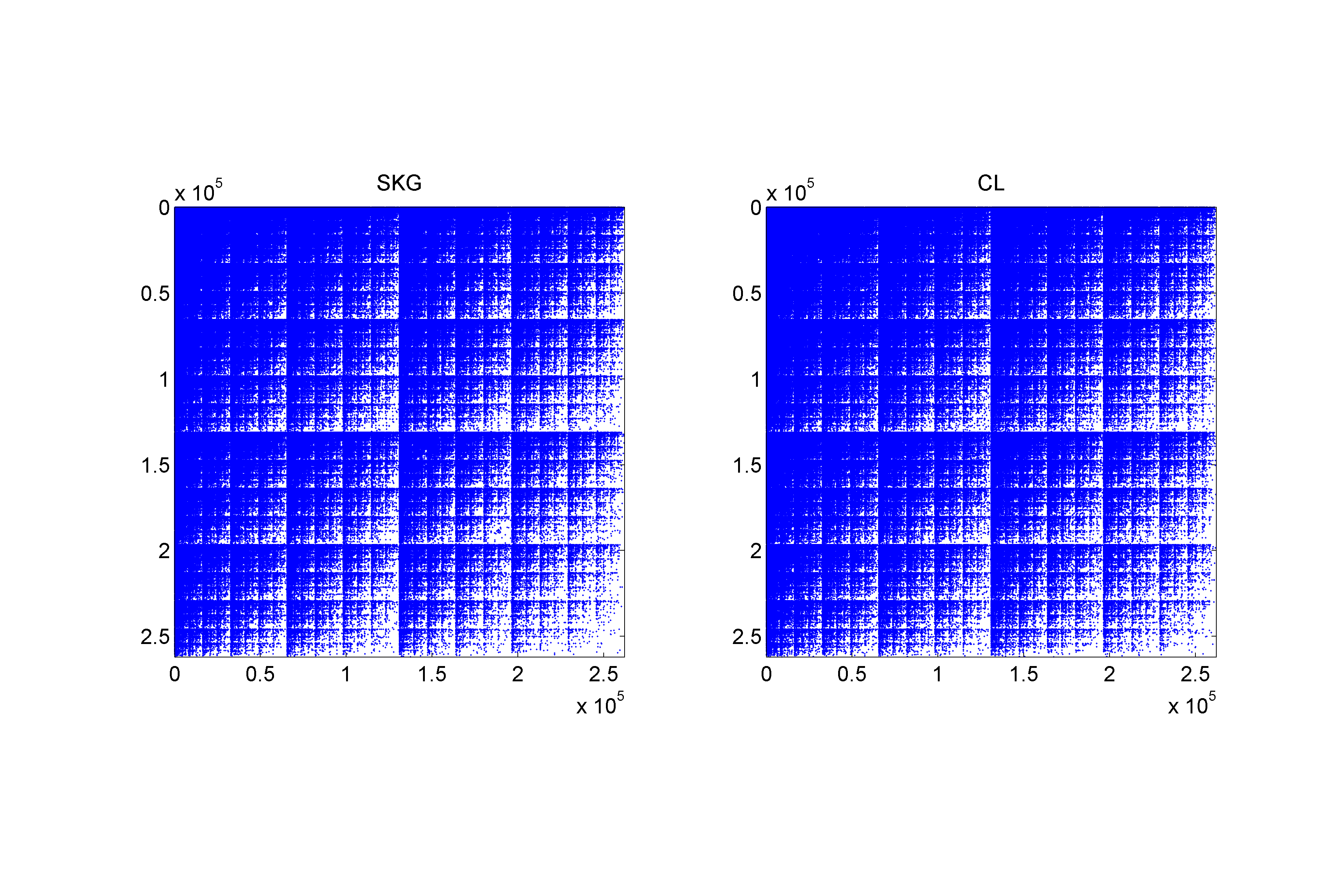} 
  \caption{Spy plots of SKG drawn from $P_{\textrm{SKG}}$ and CL graph drawn from $P_{\textrm{CL}}$}  
  \label{fig:spy}
\end{figure*}

\section{Mathematical justifications} \label{sec:math}

We prove that when the entries of the matrix $T$ satisfy the 
condition $t_1/t_2 = t_3/t_4$, then SKG is \emph{identical} to the CL model.

\medskip

\begin{theorem} \label{thm:skg-CL} Consider an SKG model where $T$ satisfies
the following:
$$ \frac{t_1}{t_2} = \frac{t_3}{t_4}. $$
Then $P_{\textrm{SKG}} = P_{\textrm{CL}}$.
\end{theorem}

\medskip

\begin{proof} Let $\alpha=t_1/t_2=t_3/t_4$, and let $t_3=\beta t_2$. Then,
$t_1=\alpha^2\beta t_4$; $t_3=\alpha\beta t_2$; and $t_2=\alpha t_4$. 
Note that since $t_1+t_2+t_3+t_4=1$, 
\begin{equation} \label{eq:sum}
(\alpha^2\beta +  \alpha + \alpha\beta+1)t_4=1
\end{equation}

We use the formula given in \Clm{entries} for the $(i,j)$ entry of the SKG and CL matrices.
 
By simple substitution, the entry for SKG is
\begin{eqnarray}
& & (t_4\alpha^2\beta)^{c_z} (t_4\alpha)^{z_i-c_z} (t_4\alpha\beta)^{z_j - c_z} {t_4}^{\ell - z_i - z_j + c_z} \nonumber \\
& = & t^\ell_4\alpha^{z_i+z_j}\beta^{z_j} \label{eq:SKGentry}
\end{eqnarray}
Analogously, for the CL matrix, the entry has value
 \begin{eqnarray*}
&&  (t_1+t_2)^{z_i}(t_3+t_4)^{\ell-z_i}(t_1+t_3)^{z_j}(t_2+t_4)^{\ell-z_j}  \\[1ex]
& = &  [t_4(\alpha^2\beta+\alpha)]^{z_i}[t_4(\alpha\beta+1)]^{\ell-z_i} \times \\
& & [t_4(\alpha^2\beta + \alpha\beta)]^{z_j}[\alpha(\alpha+1)]^{\ell-z_j}  \\[1ex]
 &=& t^{2\ell}_4 (\alpha^2\beta\! +\! \alpha)^{z_i} (\alpha\beta\!+\!1)^{\ell-z_i} (\alpha^2\beta\! +\! \alpha\beta)^{z_j} (\alpha\!+\! 1)^{\ell-z_j}  \\[1ex]
& =& t^{2\ell}_4 \alpha^{z_i+z_j}\beta^{z_j}(\alpha\beta+1)^{z_i}(\alpha\beta+1)^{\ell-z_i} \times \\
& & (\alpha+1)^{z_j} (\alpha+1)^{\ell-z_j} \\[1ex]
& =& t^{2\ell}_4 \alpha^{z_i+z_j}\beta^{z_j}(\alpha^2\beta + \alpha + \alpha\beta+1)^{\ell} \\[1ex]
& =& t^{\ell}_4 \alpha^{z_i+z_j}\beta^{z_j}[t_4(\alpha^2\beta + \alpha + \alpha\beta+1)]^{\ell} \\[1ex]
& = & t^\ell_4\alpha^{z_i+z_j}\beta^{z_j}
 \end{eqnarray*}
The last part follows from \Eqn{sum}. This is exactly the same as \Eqn{SKGentry}.
$\myproofend$
\end{proof}
	
\begin{figure*}[thb]
  \centering
  \subfloat[Degree distribution]{\label{fig:epinions-deg}
  \includegraphics[width=.9\columnwidth,trim=0 0 0 0]{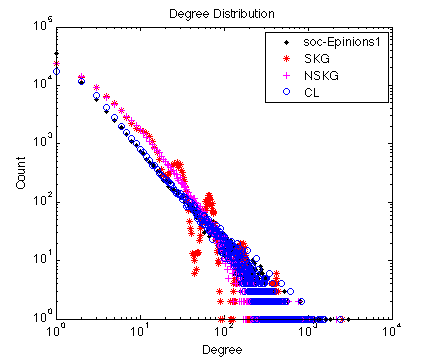}
  }
  \subfloat[Clustering coefficients]{\label{fig:epinions-cc}
  \includegraphics[width=.9\columnwidth,trim=0 0 0 0]{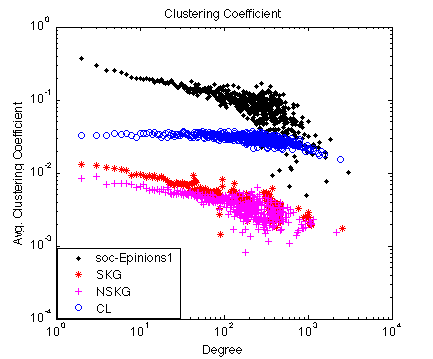}
  }\\
  \subfloat[Eigenvalues of adjacency matrices]{\label{fig:epinions-eig}
  \includegraphics[width=.9\columnwidth,trim=0 0 0 0]{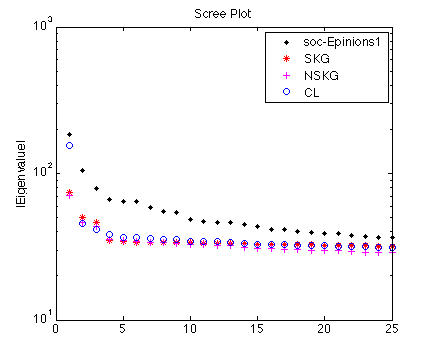}
  }
  \subfloat[Core decompositions]{\label{fig:epinions-core}
  \includegraphics[width=.9\columnwidth,trim=0 0 0 0]{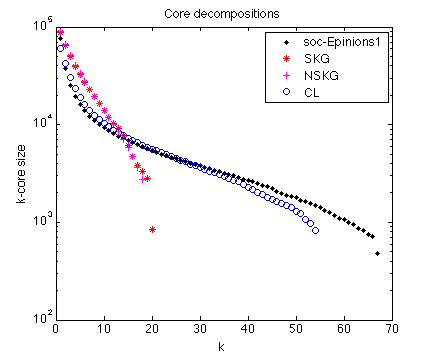}
  }  
  \caption{
  The figure compares the fits of various models for the social network soc-Epinions.}
  \label{fig:epinions}
\end{figure*}

\section{Fitting SKG vs CL} \label{sec:fit}

Fitting procedures for SKG model have been given in \cite{LeChKlFa10}. This is often cited as
a reason for the popularity of SKG. These fits are based on algorithms
for maximizing likelihood, but can take a significant amount of time to run.
The CL model is fit by simply taking the degree distribution of the original graph.
Note that the CL model uses a lot more parameters than SKG, which only requires
$5$ independent numbers. In that sense, SKG is a very appealing model regardless of any
other deficiencies.

We show comparisons of the CL, SKG, and NSKG models with respect to three different real graphs.
For directed graphs, we look at the undirected version where directions are removed from all the edges. 
The real graphs are the following:
\begin{asparaitem}
	\item soc-Epinions: This is a social network from the Epinions website, which tracks
	the ``who-trusts-whom" relationship \cite{Snap}. It has 75879 vertices and 811480 edges. The SKG
	parameters for this graph from \cite{LeChKlFa10} are:
	$T = [0.4668 \ 0.2486; 0.2243 \ 0.0603]$, $\ell = 17$.
	\item ca-HepTh: This is a co-authorship network from high energy physics \cite{Snap}. It has 9875 vertices
	and 51946 edges. The SKG parameters for this graph from \cite{LeChKlFa10} are: 
	$T =  [0.469455 \ 0.127350; 0.127350 \ 0.275846]$, $\ell = 14$.
	\item cit-HepPh: This is a citation network from high energy physics \cite{Snap}. 
	It has 34546 vertices and 841754 edges.
	The SKG parameters from \cite{LeChKlFa10} are:
	$T = [0.429559 \ 0.189715; 0.153414 \ 0.227312]$, $\ell = 15$.
\end{asparaitem}	

The comparisons between the properties are given, respectively, in \Fig{epinions}, \Fig{ca-HepTh}, and
\Fig{cit-HepPh}. In all of these,
we see that CL (as expected) gives good fits to the degree distributions. 
For soc-Epinions, we see in \Fig{epinions-deg} that the oscillations of the SKG degree distribution and how NSKG
smoothens it out.
Observe that the
clustering coefficients of all the models are completely off. Indeed, for low degree vertices,
the values are off by orders of magnitude. Clearly, no model is capturing the abundance of
triangles in these graphs. The eigenvalues of the model graphs are also distant
from the real graph, but CL performs no worse than SKG (or NSKG). Core decompositions
for soc-Epinions (\Fig{epinions-core}) show that CL fits rather well.
For ca-HepPh (\Fig{ca-HepTh-core}) CL is marginally better than SKG, whereas
for cit-HepTh (\Fig{cit-HepPh-core}), NSKG seems be a better match.

\begin{figure*}[thb]
  \centering
  \subfloat[Degree distribution]{\label{fig:ca-HepTh-deg}
  \includegraphics[width=.9\columnwidth,trim=0 0 0 0]{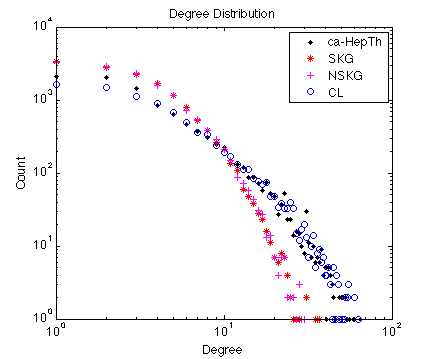}
  }
  \subfloat[Clustering coefficients]{\label{fig:ca-HepTh-cc}
  \includegraphics[width=.9\columnwidth,trim=0 0 0 0]{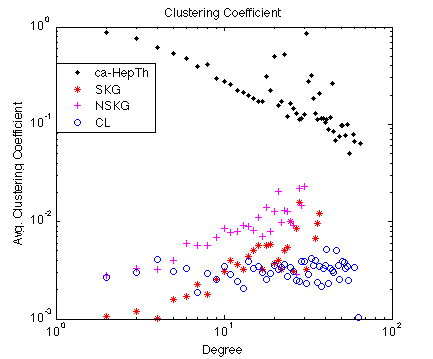}
  }\\
  \subfloat[Eigenvalues of adjacency matrices]{\label{fig:ca-HepTh-eig}
  \includegraphics[width=.9\columnwidth,trim=0 0 0 0]{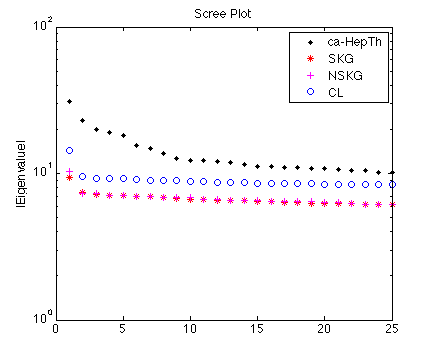}
  }
  \subfloat[Core decompositions]{\label{fig:ca-HepTh-core}
  \includegraphics[width=.9\columnwidth,trim=0 0 0 0]{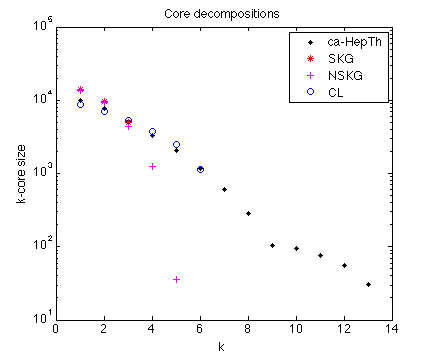}
  }  
  \caption{
  The figures compares the fits of various models for the co-authorship network ca-HepTh.}
  \label{fig:ca-HepTh}
\end{figure*}

All in all, there is no conclusive evidence that SKG or NSKG model these graphs significantly
better than CL. We feel that the comparable performance of CL shows that it should be used
as a control model to compare against.

\begin{figure*}[thb]
  \centering
  \subfloat[Degree distribution]{\label{fig:cit-HepPh-deg}
  \includegraphics[width=.9\columnwidth,trim=0 0 0 0]{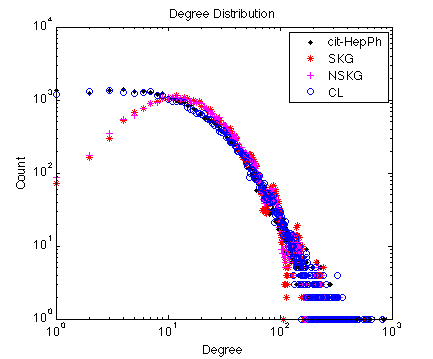}
  }
  \subfloat[Clustering coefficients]{\label{fig:cit-HepPh-cc}
  \includegraphics[width=.9\columnwidth,trim=0 0 0 0]{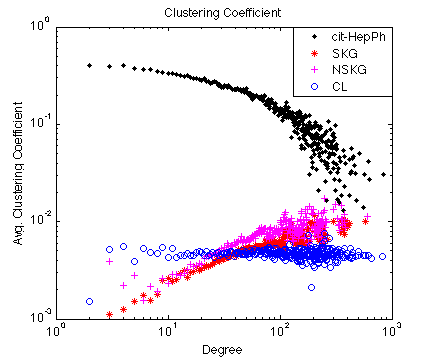}
  }\\
  \subfloat[Eigenvalues of adjacency matrices]{\label{fig:cit-HepPh-eig}
  \includegraphics[width=.9\columnwidth,trim=0 0 0 0]{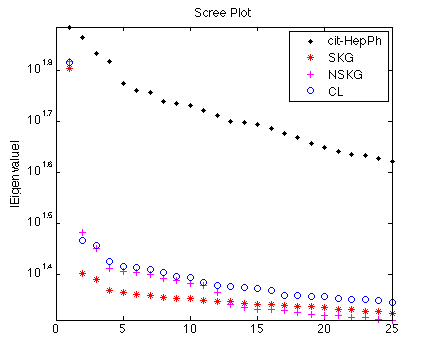}
  }
  \subfloat[Core decompositions]{\label{fig:cit-HepPh-core}
  \includegraphics[width=.9\columnwidth,trim=0 0 0 0]{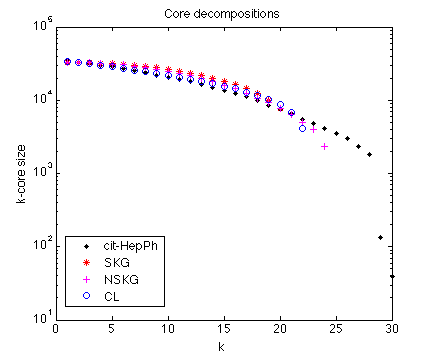}
  }  
  \caption{
  The figures compares the fits of various models for the citation network cit-HepPh.}
  \label{fig:cit-HepPh}
\end{figure*}

\section{Conclusions} \label{sec:conclusion}

Understanding existing graph models is a very important part of graph analysis. We need to
clearly see the benefits and shortcomings of existing models, so that we can use them
more effectively. For these purposes, it is good to have a simple ``baseline" model to compare
against. We feel that the CL model is quite suited for this because of its
efficiency, simplicity, and similarity to SKG. Especially for benchmarking purposes, it is a good candidate for generating simple test graphs. One should not think of this
as representing real data, but as an easy way of creating reasonable looking graphs.
Comparisons with the CL model can give more insight into current models. The similarities and differences may help identify how current graph models differ from each other.

\bibliographystyle{IEEEtran}
\bibliography{final}

\end{document}